\documentclass[runningheads]{llncs}
\usepackage[utf8]{inputenc}
\usepackage{amsmath}
\usepackage{amssymb}
\usepackage{graphicx}
\usepackage{tabularx}
\usepackage{array}
\usepackage{pifont}
\usepackage{xspace}
\usepackage{todonotes}
\usepackage{paralist}
\usepackage{thm-restate}
\usepackage{subfigure}
\usepackage[pdfpagelabels,colorlinks,citecolor=blue,linkcolor=blue,urlcolor=blue]{hyperref}

\newif \ifArxiv
\Arxivtrue

\let\doendproof\endproof
\renewcommand\endproof{~\hfill$\qed$\doendproof}
\newenvironment{sketch}{\noindent{\itshape Sketch of proof.}}{~\hfill$\qed$\doendproof\smallskip}

\spnewtheorem{clm}{Claim}{\bfseries}{\rmfamily}

\newcommand{\myparagraph}[1]{\smallskip\noindent\textbf{\boldmath #1}}
\newcommand{\smooth}{\ensuremath{\mathrm{smooth}}}
\newcommand{\pre}{\textrm{prec}}
\newcommand{\suc}{\textrm{succ}}
\newcommand{\turn}{\textrm{turn}}
\newcommand{\rot}{\textrm{rot}}
\newcommand{\rect}{\overline}

\newcommand{\lupward}{upward\xspace}
\newcommand{\ldownward}{downward\xspace}
\newcommand{\lleftward}{leftward\xspace}
\newcommand{\lrightward}{rightward\xspace}

\begin{document}

\title{
On Turn-Regular Orthogonal Representations\thanks{Work partially supported by DFG grants Ka~812/17-1 and Ru 1903/3-1; by MIUR Project ``MODE'' under PRIN 20157EFM5C; by MIUR Project ``AHeAD'' under PRIN 20174LF3T8; and by Roma Tre University Azione 4 Project ``GeoView''. This work started at the Bertinoro Workshop on Graph Drawing BWGD 2019.
\ifArxiv
Appears in the Proceedings of the 28th International Symposium on Graph Drawing and Network Visualization (GD 2020)
\else 
\fi
}
}

\ifArxiv

\author{%
Michael~A.~Bekos\inst{1} \and 
Carla~Binucci\inst{2} \and 
Giuseppe~Di~Battista\inst{3} \and
Walter~Didimo\inst{2} \and 
Martin~Gronemann\inst{4} \and 
Karsten~Klein\inst{5} \and 
Maurizio~Patrignani\inst{3} \and 
Ignaz Rutter\inst{6}}

\else
\author{%
Michael~A.~Bekos\orcidID{0000-0002-3414-7444}\inst{1} \and 
Carla~Binucci\orcidID{0000-0002-5320-9110}\inst{2} \and 
Giuseppe~Di~Battista\orcidID{0000-0003-4224-1550}\inst{3} \and
Walter~Didimo\orcidID{0000-0002-4379-6059}\inst{2} \and 
Martin~Gronemann\orcidID{0000-0003-2565-090X}\inst{4} \and 
Karsten~Klein\orcidID{0000-0002-8345-5806}\inst{5} \and 
Maurizio~Patrignani\orcidID{0000-0001-9806-7411}\inst{3} \and 
Ignaz Rutter\orcidID{0000-0002-3794-4406}\inst{6}}
\fi

\authorrunning{M.~A.~Bekos et al.}

\date{}

\institute{%
Department of Computer Science, University of T{\"u}bingen, T{\"u}bingen, Germany
\\\email{bekos@informatik.uni-tuebingen.de}
\and
Department of Engineering, University of Perugia, Perugia, Italy\\
\email{carla.binucci@unipg.it},
\email{walter.didimo@unipg.it}
\and
Department of Engineering, Roma Tre University, Italy\\
\email{gdb@dia.uniroma3.it},  
\email{maurizio.patrignani@uniroma3.it}
\and
Theoretical Computer Science, Osnabr\"uck University, Osnabr\"uck, Germany
\\\email{martin.gronemann@uni-osnabrueck.de}
\and
Department of Computer and Information Science, University of Konstanz, Konstanz, Germany
\email{karsten.klein@uni-konstanz.de}
\and
Department of Computer Science and Mathematics, University of Passau, Germany
\email{rutter@fim.uni-passau.de}
}

\maketitle
%
\begin{abstract}
An interesting class of orthogonal representations consists of the so-called \emph{turn-regular} ones, i.e., those that do not contain any pair of reflex corners that ``point to each other'' inside a face. For such a representation $H$ it is possible to compute in linear time a minimum-area drawing, i.e., a drawing of minimum area over all possible assignments of vertex and bend coordinates of $H$. In contrast, finding a minimum-area drawing of $H$ is NP-hard if $H$ is non-turn-regular. This scenario naturally motivates the study of which graphs admit turn-regular orthogonal representations. In this paper we identify notable classes of biconnected planar graphs that always admit such  representations, which can be computed in linear time. We also describe a linear-time testing algorithm for trees and provide a polynomial-time algorithm that tests whether a biconnected plane graph with ``small'' faces has a turn-regular orthogonal representation without bends.  
%
\end{abstract}

\section{Introduction}\label{sec:intro}

Computing \emph{orthogonal drawings} of graphs is among the most studied problems in graph drawing~\cite{DBLP:books/ph/BattistaETT99,DBLP:reference/crc/DuncanG13,DBLP:conf/dagstuhl/1999dg,DBLP:books/ws/NishizekiR04}, because of its direct application to several domains, such as software engineering, information systems, and circuit design (e.g.,~\cite{DBLP:journals/jss/BatiniTT84,dl-gvdm-07,DBLP:journals/ivs/EiglspergerGKKJLKMS04,DBLP:books/sp/Juenger04,Lengauer-90}). In an orthogonal drawing, the vertices of the graph are mapped to distinct points of the plane and each edge is represented as~an alternating sequence of horizontal and vertical segments between its end-vertices. A point in which two segments of an edge meet is called~a~\emph{bend}. An orthogonal drawing is a \emph{grid} drawing if its vertices and bends have integer coordinates.

One of the most popular and effective strategies to compute a readable orthogonal grid drawing of a graph $G$ is the so-called \emph{topology-shape-metrics} (or \emph{TSM}, for short) approach~\cite{DBLP:journals/siamcomp/Tamassia87}, which consists of three 
steps: 
\begin{inparaenum}[(i)]
\item \label{s:1} compute a planar embedding of $G$ by possibly adding dummy vertices to replace edge crossings if $G$ is not planar; 
\item \label{s:2} obtain an \emph{orthogonal representation} $H$ of $G$ from the previously determined planar embedding; $H$ describes the ``shape'' of the final drawing in terms of angles around the vertices and sequences of left/right bends along the edges; 
\item \label{s:3} assign integer coordinates to vertices and bends of $H$ to obtain the final non-crossing orthogonal grid drawing~$\Gamma$~of~$G$.
\end{inparaenum}

If $G$ is planar, the TSM approach computes a planar  orthogonal grid drawing $\Gamma$ of $G$. Such a planar drawing exists if and only if $G$ is a \emph{4-graph}, i.e., of maximum vertex-degree at most four. To increase the readability of $\Gamma$, a typical optimization goal of Step~(\ref{s:2}) is the minimization of the number of bends. In Step~(\ref{s:3}) the goal is to minimize the area or the total edge length of $\Gamma$; a problem referred to as \emph{orthogonal compaction}. Unfortunately, while the computation of an embedding-preserving bend-minimum orthogonal representation $H$ of a plane 4-graph is polynomial-time solvable~\cite{DBLP:journals/jgaa/CornelsenK12,DBLP:journals/siamcomp/Tamassia87}, the orthogonal compaction problem for a planar orthogonal representation $H$ is NP-complete in the general case~\cite{DBLP:journals/comgeo/Patrignani01}. 
Nevertheless, Bridgeman et al.~\cite{DBLP:journals/comgeo/BridgemanBDLTV00} showed that the compaction problem for the area requirement can be solved optimally in linear time for a subclass of orthogonal representations called \emph{turn-regular}. A similar polynomial-time result for the minimization of the total edge length in this case is proved by Klau and Mutzel~\cite{DBLP:conf/ipco/KlauM99}. Esser showed that these two approaches are equivalent~\cite{10.1007/978-3-030-41590-7_8}. 

Informally speaking, a face of a planar orthogonal representation $H$ is turn-regular if it does not contain a pair of reflex corners (i.e., turns of $270^\circ$) that point to each other (see Section~\ref{sec:preliminaries} for the formal definition); $H$ is turn-regular if all its faces are turn-regular. For a turn-regular representation $H$, every pair of vertices or bends has a unique orthogonal relation (left/right or above/below) in any planar drawing of $H$. Conversely, different orthogonal relations are possible for a pair of opposing reflex corners, which makes it computationally hard to optimally compact non-turn-regular representations. For example, Figs.~\ref{fig:intro-1} and~\ref{fig:intro-2} show two different drawings of a non-turn-regular orthogonal representation; the drawing in Fig.~\ref{fig:intro-2} has minimum area. Fig.~\ref{fig:intro-3} depicts a minimum-area drawing of a turn-regular orthogonal representation of the same graph.

\begin{figure}[tb]
	\centering
	\subfigure[]{%
		\centering
		\includegraphics[page=1]{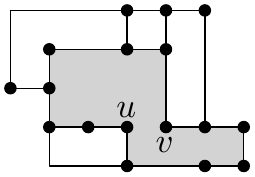}
		\label{fig:intro-1}
	}\hfil
	\subfigure[]{%
		\centering
		\includegraphics[page=2]{figures/intro.pdf}
		\label{fig:intro-2}
	}\hfil
	\subfigure[]{%
		\centering
		\includegraphics[page=3]{figures/intro.pdf}
		\label{fig:intro-3}
	}
	\caption{%
	(a) Drawing of a non-turn-regular orthogonal representation $H$; vertices $u$ and $v$ point to each other in the gray face. 
	(b) Another drawing of $H$ with smaller area. 
	(c) Drawing of a turn-regular orthogonal representation of the~same~graph.}
	\label{fi:intro}
\end{figure}

The aforementioned scenario naturally motivates the problem of computing orthogonal representations that are turn-regular, so to support their subsequent compaction. To the best of our knowledge, this problem has not been studied so far (a related problem is studied for upward planar drawings only~\cite{DBLP:journals/jgaa/BinucciGDR11,DBLP:conf/gd/BattistaL98,DBLP:journals/jgaa/Didimo06}). Heuristics have been described to make any given orthogonal representation $H$ turn-regular, by adding a minimal set of dummy edges~\cite{DBLP:journals/comgeo/BridgemanBDLTV00,DBLP:journals/amc/HashemiT06}; however, such edges impose constraints that may yield a drawing of sub-optimal area for~$H$. 

\smallskip\noindent Our contribution is as follows:

\smallskip\noindent $(i)$ We identify notable classes of planar graphs that always admit turn-regular orthogonal representations. We prove that biconnected planar 3-graphs and planar Hamiltonian 4-graphs (which include planar 4-connected 4-graphs~\cite{NC08}) admit turn-regular representations with at most two bends per edge and at most three bends per edge, respectively. For these graphs, a turn-regular representation can be constructed in linear time. We also prove that every biconnected planar graph admits an orthogonal representation that is \emph{internally} turn-regular, i.e., its internal faces are turn-regular (Section~\ref{sec:turn-regular-graphs}). We leave open the question whether every biconnected planar 4-graph admits a turn-regular representation. 

\smallskip\noindent $(ii)$ For 1-connected planar graphs, including trees, there exist infinitely many instances for which a turn-regular representation does not exist. Motivated by this scenario, and since the orthogonal compaction problem remains NP-hard even for orthogonal representations of 
paths~\cite{efkssw-mrpgas-20-arxiv},
we study and characterize the class of trees that admit turn-regular representations. We then describe a corresponding linear-time testing algorithm, which in the positive case computes a turn-regular drawing without bends (Section~\ref{sec:testing-turn-regularity}). Finally, we prove that such drawings are ``convex'' (i.e., all edges incident to leaves can be extended to infinite crossing-free rays). We remark that a linear-time algorithm to compute planar straight-line convex drawings of trees is described by Carlson and Eppstein~\cite{DBLP:conf/gd/CarlsonE06}. However, in general, the drawings they compute are not orthogonal.

\smallskip\noindent $(iii)$ We address the problem of testing whether a given biconnected plane graph admits a turn-regular \emph{rectilinear} representation, i.e., a representation without bends. For this problem we give a polynomial-time algorithm for plane graphs with ``small'' faces, namely faces of degree at most eight (Section~\ref{sec:rect}).

	

\section{Preliminary Definitions and Basic Results}\label{sec:preliminaries}

We consider connected graphs and assume familiarity with basic concepts of orthogonal graph drawing and planarity~\cite{DBLP:books/ph/BattistaETT99} 
\ifArxiv
(see Appendix~\ref{app:preliminaries}). 
\else
(see also~\cite{bbddgkpr-tror-20-arxiv}).
\fi
Let $G$ be a plane 4-graph and $H$ be an orthogonal representation of $G$. If $H$ has no bends, then it is called \emph{rectilinear}. W.l.o.g., we assume that $H$ comes with a given orientation, i.e., for each edge segment $\rect{pq}$ of $H$ (where $p$ and $q$ are vertices or bends), it is fixed if $p$ is to the left, to the right, above, or below $q$ in every (orthogonal) drawing of $H$. Let $f$ be a face of $H$. We assume that the boundary of $f$ is traversed counterclockwise (clockwise) if $f$ is internal (external). The \emph{rectilinear image} of $H$ is the orthogonal representation $\rect{H}$ obtained from $H$ by replacing each bend with a degree-2 vertex. For any face $f$ of $H$, let $\rect{f}$ denote the corresponding face of $\rect{H}$. For each occurrence of a vertex $v$ of $\rect{H}$ on the boundary of $\rect{f}$, let $\pre(v)$ and $\suc(v)$ be the edges preceding and following $v$, respectively, on the boundary of $\rect{f}$ ($\pre(v)=\suc(v)$ if $\deg(v)=1$). Let $\alpha$ be the value of the angle internal to $\rect{f}$ between $\pre(v)$ and $\suc(v)$. We associate with $v$ one or two \emph{corners} based on the following cases: If $\alpha=90^\circ$, associate with $v$ one \emph{convex} corner; if $\alpha=180^\circ$, associate with $v$ one \emph{flat} corner; if $\alpha=270^\circ$, associate with $v$ one \emph{reflex} corner; if $\alpha=360^\circ$, associate with $v$ an ordered pair of \emph{reflex} corners. For example, in the (internal) face of Fig.~\ref{fig:turn}, a convex corner is associated with $v_1$, a flat corner with $v_2$, a reflex corner with $v_3$, and an ordered pair of reflex corners with $v_4$. 

\begin{figure}[tb]
	\centering
	\subfigure[]{%
		\centering
		\includegraphics{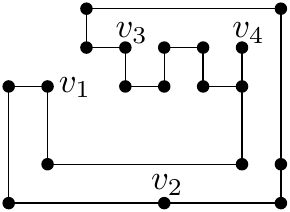}\label{fig:turn}
	}\hfil
	\subfigure[]{%
		\centering
		\includegraphics{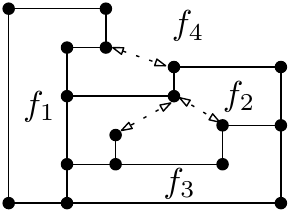}\label{fig:kitty}
	}
	\caption{Illustration of 
	(a)~convex, flat and reflex corners, and
	(b)~kitty corners.}
	\label{fig:forbidden-trees}
\end{figure}

Based on the definition above, there is a circular sequence of corners associated with (the boundary) of $\rect{f}$. For a corner $c$ of $\rect{f}$, we define: $\turn(c)= 1$ if $c$ is convex; $\turn(c)= 0$ if $c$ is flat; $\turn(c)= -1$ if $c$ is reflex. For any ordered pair $(c_i, c_j)$ of corners of $\rect{f}$, we define the following function: $\rot(c_i,c_j)=\sum_c \turn(c)$ for all corners $c$ along the boundary of $\rect{f}$ from $c_i$ (included) to $c_j$ (excluded). For example, in Fig.~\ref{fig:turn} let $c_1$, $c_2$, and $c_3$ be the corners associated with $v_1$, $v_2$, and $v_3$, respectively, and let $(c_4,c'_4)$ be the ordered pair of reflex corners associated with~$v_4$. We have $\rot(c_1,c_2)=3$, $\rot(c_3,c_4)=1$, $\rot(c_3,c'_4)=0$, and $\rot(c_3,c_1)=-3$. The properties below are consequences of results in~\cite{DBLP:journals/siamcomp/Tamassia87,DBLP:journals/siamcomp/VijayanW85}.
%
%
\begin{property}\label{pr:rot-1}
For each face $\rect{f}$ of $\rect{H}$ and for each corner $c_i$ of $\rect{f}$, we have $\rot(c_i,c_i)=4$ if $\rect{f}$ is internal and $\rot(c_i,c_i)=-4$ if $\rect{f}$ is external.
\end{property}	 
\begin{property}\label{pr:rot-2}
For each ordered triplet of corners $(c_i, c_j, c_k)$ of a face of $\rect{H}$, we have $\rot(c_i,c_k)=\rot(c_i,c_j) + \rot(c_j,c_k)$.  
\end{property}
\begin{property}\label{pr:rot-3}
Let $c_i$ and $c_j$ be two corners of $\rect{f}$. If $\rect{f}$ is internal then $\rot(c_i,c_j)=2$ $\iff$ $\rot(c_j,c_i)=2$. If $\rect{f}$ is external then $\rot(c_i,c_j)=2$ $\iff$ $\rot(c_j,c_i)=-6$.
\end{property}
%

Let $c$ be a reflex corner of $\rect{H}$ associated with either a degree-2 vertex or a bend of $H$. Let $s_h$ and $s_v$ be the horizontal and vertical segments incident to $c$ and let $\ell_h$ and $\ell_v$ be the lines containing $s_h$ and $s_v$, respectively.  We say that $c$ (or equivalently its associated vertex/bend of $H$) points \emph{up-left}, if $s_h$ is to the right of $\ell_v$ and $s_v$ is below $\ell_h$. The definitions of $c$ that points \emph{up-right}, \emph{down-left}, and \emph{down-right} are symmetric (see Figs.~\ref{fig:degree-2-corners-up-left}-\ref{fig:degree-2-corners-down-right}). If $v$ is a degree-1 vertex in $H$, then it has two associated reflex corners in $\rect{H}$. In this case, $v$ points \emph{\lupward} (\emph{\ldownward}) if its incident segment is  vertical and below (above) the horizontal line passing through $v$. The definitions of a degree-1 vertex that points \emph{\lleftward} or \emph{\lrightward} are symmetric (see Figs.~\ref{fig:degree-1-corners-up-left-right}-\ref{fig:degree-1-corners-up-down-right}). 

\begin{figure}[tb]
	\centering
	\subfigure[]{%
		\includegraphics[page=1]{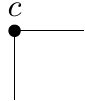}
		\label{fig:degree-2-corners-up-left}
	}\hfill
	\subfigure[]{%
		\includegraphics[page=2]{figures/degree-2-corners.pdf}
		\label{fig:degree-2-corners-up-right}
	}\hfill
	\subfigure[]{%
		\includegraphics[page=3]{figures/degree-2-corners.pdf}
		\label{fig:degree-2-corners-down-left}
	}\hfill
	\subfigure[]{%
		\includegraphics[page=4]{figures/degree-2-corners.pdf}
		\label{fig:degree-2-corners-down-right}
	}\hfill
	\subfigure[]{%
		\includegraphics[page=1]{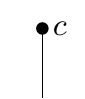}
		\label{fig:degree-1-corners-up-left-right}
	}\hfill
	\subfigure[]{%
		\includegraphics[page=2]{figures/degree-1-corners.pdf}
		\label{fig:degree-1-corners-down-left-right}
	}\hfill
	\subfigure[]{%
		\includegraphics[page=3]{figures/degree-1-corners.pdf}
		\label{fig:degree-1-corners-up-down-left}
	}\hfill
	\subfigure[]{%
		\includegraphics[page=4]{figures/degree-1-corners.pdf}
		\label{fig:degree-1-corners-up-down-right}
	}
	\caption{%
	Directions of a reflex corner associated with a degree-2 vertex or with a bend: 
	(a) up-left; (b) up-right; (c) down-left; (d) down-right. 
	Directions of a degree-1 vertex: 
	(e) \lupward; (f) \ldownward; (g) \lleftward; (h) \lrightward.}
	\label{fig:coner-directions}
\end{figure}
	 
Two reflex corners $c_i$ and $c_j$ of a face of $\rect{H}$ are called \emph{kitty corners} if $\rot(c_i,c_j)=2$ or $\rot(c_j,c_i)=2$. A face $f$ of an orthogonal representation $H$ is \emph{turn-regular}, if the corresponding face $\rect{f}$ of $\rect{H}$ has no kitty corners. If every face of $H$ is turn-regular, then $H$ is \emph{turn-regular}. For example, the orthogonal representation in Fig.~\ref{fig:kitty} is not turn-regular as the faces $f_1$ and $f_3$ are turn-regular, while the internal face $f_2$ and the external face $f_4$ are not turn-regular (the pairs of kitty corners in each face are highlighted with dotted arrows). A graph $G$ is \emph{turn-regular}, if it admits a turn-regular orthogonal representation. If $G$ admits a turn-regular rectilinear representation, then $G$ is \emph{rectilinear turn-regular}. The next lemma 
\ifArxiv
(whose proof can be found in Appendix~\ref{app:preliminaries}), 
\else
(whose proof can be found in~\cite{bbddgkpr-tror-20-arxiv}), 
\fi
provides a sufficient condition for the existence of a kitty-corner pair in the external face.

\begin{restatable}{lemma}{kittycorners}\label{lem:outer-face-kitty-corners}
  Let $\overline H$ be the rectilinear image of an orthogonal
  representation $H$ of a plane graph $G$.  Let $(c_1,c_2)$ be two corners of the external face such that
  $\rot(c_1,c_2) \ge 3$ or $c_1$ is a reflex corner
  and~$\rot(c_1,c_2) \ge 2$.  Then, the external face contains a pair of kitty corners.
\end{restatable}

\begin{corollary}\label{cor:outer-face-three-convex}
   Let $H$ be an orthogonal representation of a plane graph $G$. If the external face of~$\overline H$ has three consecutive convex corners,
   $H$ is not turn-regular.
\end{corollary}

\section{Turn-Regular Graphs}\label{sec:turn-regular-graphs}

The theorems in this section can be proven by modifying a well-known linear-time algorithm by Biedl and Kant~\cite{DBLP:journals/comgeo/BiedlK98} that produces an orthogonal drawing $\Gamma$ with at most two bends per edge of a biconnected planar 4-graph~$G$ with a fixed embedding~$\mathcal{E}$. Such an algorithm exploits an $st$-ordering $s = v_1, v_2, \dots, v_n = t$ of the vertices of $G$, where $s$ and $t$ are two distinct vertices on the external face of $\mathcal{E}$. We recall that an $st$-ordering $s = v_1, v_2, \dots, v_n = t$ is a linear ordering of the vertices of $G$ such that any vertex $v_i$ distinct from $s$ and $t$ has at least two neighbors $v_j$ and $v_k$ in $G$ with $j < i < k$~\cite{DBLP:journals/tcs/EvenT77}.
The orthogonal drawing $\Gamma$ is constructed incrementally by adding vertex $v_k$, for $k = 1, \dots, n$, into the drawing $\Gamma_{k-1}$ of $\{v_1, \dots, v_{k-1}\}$, while preserving the embedding  $\mathcal{E}$. Some invariants are maintained when vertex $v_k$ is placed above $\Gamma_{k-1}$: (i) vertex $v_k$ is attached to $\Gamma_{k-1}$ with at least one edge incident to $v_k$ from the bottom; (ii) after $v_k$ is added to $\Gamma_{k-1}$, some extra columns are introduced into $\Gamma_{k}$ to ensure that each edge $(v_i,v_j)$, such that $i \leq k < j$ has a dedicated column in $\Gamma_{k}$ that is reachable from $v_i$ with at most one bend and without introducing crossings.  

\begin{figure}[tb]
  \centering
  \subfigure[]{%
    \centering
    \includegraphics[page=1]{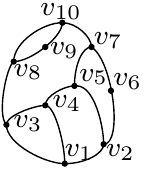}
    \label{fig:example-biconn-deg3}
  }\hfil
  \subfigure[]{%
    \centering
    \includegraphics[page=2]{figures/biconn-deg3.pdf}
    \label{fig:example-biconn-deg3-v1}
  }\hfil
  \subfigure[]{%
    \centering
    \includegraphics[page=3]{figures/biconn-deg3.pdf}
    \label{fig:example-biconn-deg3-v2}
  }\hfil
  \subfigure[]{%
    \centering
    \includegraphics[page=4]{figures/biconn-deg3.pdf}
    \label{fig:example-biconn-deg3-v3}
  }\hfil
  \subfigure[]{%
    \centering
    \includegraphics[page=5,width=0.24\textwidth]{figures/biconn-deg3.pdf}
    \label{fig:example-biconn-deg3-v4}
  }
  \caption{The first four steps of the algorithm in the proof of Theorem~\ref{th:biconn-deg3} for the construction of a turn-regular orthogonal drawing of the biconnected planar 3-graph 
  \ifArxiv
  shown in~(a) (the following steps are illustrated in Fig.~\ref{fig:biconn-deg3-whole} in Appendix~\ref{app:turn-regular}).
  \else
  shown in~(a).
  \fi}
  \label{fig:biconn-deg3}
\end{figure}

\begin{theorem}\label{th:biconn-deg3}
Every biconnected planar 3-graph admits a turn-regular representation with at most two bends per edge, which can be computed in linear time.
\end{theorem}
\begin{proof}
%
Let $G$ be a biconnected planar 3-graph and let $\mathcal{E}$ be any planar embedding of $G$. Let $s$ and $t$ be two distinct vertices on the external face of $\mathcal{E}$. As in~\cite{DBLP:journals/comgeo/BiedlK98}, based on an $st$-ordering of $G$, we incrementally construct an orthogonal drawing $\Gamma$ of $G$ by adding $v_k$ into the drawing $\Gamma_{k-1}$ of graph $G_{k-1}$ , for $k = 1, \dots, n$ . Besides the invariants (i) and (ii) described above, 
we additionally maintain the invariant (iii): each reflex corner introduced in the drawing points either down-right or up-right with the possible exception of the reflex corners of the edges on the external face that are incident to~$s$ or~$t$.
Drawing $\Gamma_1$ consists of the single vertex $v_1$. Since $\deg(v_1) \leq 3$, the columns assigned to its three incident edges are the column where $v_1$ lies and the two columns immediately on its left and on its right (see Fig.~\ref{fig:example-biconn-deg3-v1}). These columns are assigned to the edges incident to $v_1$ in the order they appear in $\mathcal{E}$. Now, suppose you have to add vertex $v_k$ to $\Gamma_{k-1}$. Observe that, since $G$ has maximum degree three, $v_k$ has a maximum of three edges $(v_k,v_h)$, $(v_k,v_i)$, and $(v_k,v_j)$, where we may assume, without loss of generality, that $h < i < j$. To complete the proof, we consider three cases:

\smallskip\noindent\emph{Case 1} ($h < k < i$): We place $v_k$ on the first empty row above $\Gamma_{k-1}$ and on the column assigned to $(v_k, v_h)$. Also, to preserve the invariant (ii), we introduce an extra column immediately to the right of $v_k$ and we assign the column of $v_k$ and the newly added extra column to $(v_k,v_i)$ and $(v_k,v_j)$ in the order that is given by $\mathcal{E}$. For example, Fig.~\ref{fig:example-biconn-deg3-v2} shows the placement of $v_2$ directly above $v_1$, with one extra column inserted to the right of $v_2$ and assigned to the edge $(v_2,v_6)$. 

\smallskip\noindent\emph{Case 2} ($i < k < j$): We place $v_k$ on the first empty row above $\Gamma_{k-1}$ and on the leftmost column between the columns assigned to $(v_k, v_h)$ and $(v_k, v_i)$. Also, we assign the column of $v_k$ to $(v_k,v_j)$, e.g., Fig.~\ref{fig:example-biconn-deg3-v4} shows the placement of $v_4$ on the leftmost column assigned to its incoming edges $(v_3,v_4)$ and $(v_1,v_4)$. 

\smallskip\noindent\emph{Case 3} ($j < k$): Here, $v_k$ is $t$. We place $v_k$ on the first empty row above $\Gamma_{k-1}$ and on the middle column among those assigned to $(v_k, v_h)$, $(v_k, v_i)$, and $(v_k, v_j)$.
     
\smallskip The discussion in~\cite{DBLP:journals/comgeo/BiedlK98} suffices to prove that $\Gamma$ is a planar orthogonal drawing of $G$ with at most two bends per edge. We claim that $\Gamma$ is also turn-regular. In fact, the invariant (iii) guarantees that all internal faces have reflex corners pointing either down-right or up-right and, hence, are turn-regular. On the external face we may have reflex corners pointing down-left (from the leftmost edge of $v_1$) or up-left (from the leftmost edge of $v_n$). However, since there is a $y$-monotonic path leading from $s$ to any other vertex of $G$, such corners correspond to bends lying on the bottom or on the top row of any drawing with the same orthogonal representation as $\Gamma$ and, therefore, they cannot form a kitty corner.      
\end{proof}

The proofs of the next two theorems exploit a similar technique as in Theorem~\ref{th:biconn-deg3}. 
The full proof of Theorem~\ref{th:hamiltonian-deg4} can be found in  
\ifArxiv
Appendix~\ref{app:turn-regular}. 
\else 
\cite{bbddgkpr-tror-20-arxiv}.
\fi 

\begin{figure}[b]
  \centering
  \subfigure[]{%
    \centering
    \includegraphics[page=1]{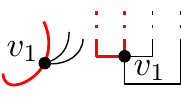}
    \label{fig:ham-first-1}
  }\hfil
  \subfigure[]{%
    \centering
    \includegraphics[page=2]{figures/ham-first.pdf}
    \label{fig:ham-first-2}
  }\hfil
  \subfigure[]{%
    \centering
    \includegraphics[page=3]{figures/ham-first.pdf}
    \label{fig:ham-first-3}
  }
  \subfigure[]{%
    \centering
    \includegraphics[page=1]{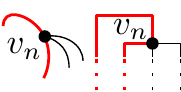}
    \label{fig:ham-last-1}
  }\hfil
  \subfigure[]{%
    \centering
    \includegraphics[page=2]{figures/ham-last.pdf}
    \label{fig:ham-last-2}
  }\hfil
  \subfigure[]{%
    \centering
    \includegraphics[page=3]{figures/ham-last.pdf}
    \label{fig:ham-last-3}
  }
  
  \subfigure[]{%
    \centering
    \includegraphics[page=1]{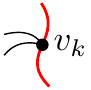}
    \label{fig:ham-unbalanced-1}
  }\hspace{0.05cm}
  \subfigure[]{%
    \centering
    \includegraphics[page=2]{figures/ham-unbalanced.pdf}
    \label{fig:ham-unbalanced-2}
  }\hspace{0.05cm}
  \subfigure[]{%
    \centering
    \includegraphics[page=3]{figures/ham-unbalanced.pdf}
    \label{fig:ham-unbalanced-3}
  }\hspace{0.05cm}
  \subfigure[]{%
    \centering
    \includegraphics[page=4]{figures/ham-unbalanced.pdf}
    \label{fig:ham-unbalanced-4}
  }\hfil\hfil
  \subfigure[]{%
    \centering
    \includegraphics[page=1]{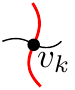}
    \label{fig:ham-balanced-1}
  }\hfil
  \subfigure[]{%
    \centering
    \includegraphics[page=2]{figures/ham-balanced.pdf}
    \label{fig:ham-balanced-2}
  }\hfil
  \subfigure[]{%
    \centering
    \includegraphics[page=3]{figures/ham-balanced.pdf}
    \label{fig:ham-balanced-3}
  }\hfil
  \subfigure[]{%
    \centering
    \includegraphics[page=4]{figures/ham-balanced.pdf}
    \label{fig:ham-balanced-4}
  }\hfil
  \subfigure[]{%
    \centering
    \includegraphics[page=5,width=0.07\textwidth]{figures/ham-balanced.pdf}
    \label{fig:ham-balanced-5}
  }
  \caption{Drawing rules for the algorithm in the proof of Theorem~\ref{th:hamiltonian-deg4}. The Hamiltonian path is drawn red and thick. Figs. (g)--(j) are to be intended up to a horizontal flip.}\label{fig:ham-deg3}
\end{figure}

\begin{restatable}{theorem}{hamiltoniandegfour}\label{th:hamiltonian-deg4}
Every planar Hamiltonian 4-graph $G$ admits a turn-regular representation $H$ with at most $3$ bends per edge, and such that only one edge of $H$ gets $3$ bends and only if $G$ is 4-regular. Given the Hamiltonian cycle, $H$ can be computed in linear~time.
\end{restatable}
\begin{sketch}
We use as $st$-ordering for the Biedl and Kant approach~\cite{DBLP:journals/comgeo/BiedlK98} the ordering given by the Hamiltonian cycle.  We choose a suitable vertex $v_1$ from which we start the construction. The construction rules are given in Fig.~\ref{fig:ham-deg3}. 
\ifArxiv
A full example is provided in Fig.~\ref{fig:hamiltonian-deg4-whole} in Appendix~\ref{app:turn-regular}.
\fi 
If $G$ has a vertex of degree less than four, then we choose such a vertex as $v_1$.
Otherwise, $G$ is 4-regular and we prove that $G$ has at least one vertex such that the configuration of Fig.~\ref{fig:ham-first-3} is ruled out by the embedding~$\mathcal{E}$.
We maintain the invariant that all the reflex corners introduced in the drawing point 
\begin{inparaenum}[(i)] \item down-left or up-left, if they are contained in a face that is on the left side of the portion of the Hamiltonian cycle traversing $\Gamma_k$, and 
\item down-right or up-right, if they are contained in a face that is on the right side. 
\end{inparaenum}
Possible exceptions are the reflex corners on the external face and that occur on edges incident to~$v_1$ or to~$v_n$. 
\end{sketch}


\begin{restatable}{theorem}{biconnectedinternally}\label{th:biconnected-internally}
  Every biconnected planar 4-graph has a representation with $O(n)$
  bends per edge that is internally turn-regular and that is computed in $O(n)$~time.
\end{restatable}
\begin{proof}
We modify the algorithm of Biedl and Kant~\cite{DBLP:journals/comgeo/BiedlK98} again, where instead of the standard bottom-up construction, we adopt a spiraling one. The vertices are inserted in the drawing according to an $st$-ordering, based on the rules depicted in 
\ifArxiv
Fig.~\ref{fig:roll} (a full example of the algorithm is in Fig.~\ref{fig:roll-whole} in Appendix~\ref{app:turn-regular}). 
\else
Fig.~\ref{fig:roll}.
\fi
For an internal face $f$ let $s(f)$ ($d(f)$, resp.) be the index of the first (last, resp.) inserted vertex incident to $f$. By construction, $f$ is bounded by two paths $P_{\ell}$ and $P_r$ that go from $v_{s(f)}$ to $v_{d(f)}$, where $P_{\ell}$ precedes $P_r$ in the left-to-right list of outgoing edges of~$v_{s(f)}$. The construction rules imply that $P_r$ has only convex corners. On the other hand, each convex corner of $P_{\ell}$ is always immediately preceded or immediately followed by a reflex corner. This rules out kitty corners in $f$. Indeed, consider to reflex corners $c_i$ and $c_j$ of $P_{\ell}$ and the counter-clockwise path from $c_i$ to $c_j$ all contained into $P_{\ell}$. When computing $\rot(c_i,c_j)$  a positive amount $+1$ is always followed by a negative amount $-1$, and the sum is never equal to $2$. Since $f$ is an internal face,  $\rot(c_i,c_j)+\rot(c_j,c_i) = 4$, and $\rot(c_i,c_j) \neq 2$ implies $\rot(c_j,c_i) \neq 2$. Note that an edge $(v_i,v_j)$ contains $O(j-i)$ bends, which yields $O(n)$ bends per edge in the worst case. 
\end{proof}

\begin{figure}[tb]
  \centering
  \subfigure[]{%
    \centering
    \includegraphics[page=1,width=0.22\textwidth]{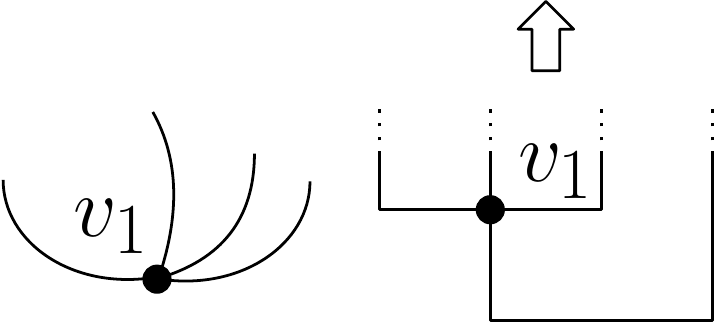}
    \label{fig:roll-1}
  }\hfil
  \subfigure[]{%
    \centering
    \includegraphics[page=2,width=0.18\textwidth]{figures/roll.pdf}
    \label{fig:roll-2}
  }\hfil
  \subfigure[]{%
    \centering
    \includegraphics[page=3,width=0.18\textwidth]{figures/roll.pdf}
    \label{fig:roll-3}
  }\hfil
  \subfigure[]{%
    \centering
    \includegraphics[page=4,width=0.18\textwidth]{figures/roll.pdf}
    \label{fig:roll-4}
  }\hfil
  
  \subfigure[]{%
    \centering
    \includegraphics[page=5,width=0.18\textwidth]{figures/roll.pdf}
    \label{fig:roll-5}
  }\hfil
  \subfigure[]{%
    \centering
    \includegraphics[page=6,width=0.15\textwidth]{figures/roll.pdf}
    \label{fig:roll-6}
  }\hfil
  \subfigure[]{%
    \centering
    \includegraphics[page=7,width=0.22\textwidth]{figures/roll.pdf}
    \label{fig:roll-7}
  }
  \caption{The construction rules for the algorithm in the proof of Theorem~\ref{th:biconnected-internally}.}\label{fig:roll}
\end{figure}

\section{Characterization \& Recognition of Turn-Regular Trees}\label{sec:testing-turn-regularity}

We give a characterization of the trees that admit turn-regular representations, which we use to derive a corresponding linear-time testing and drawing algorithm. 
For a tree $T$, let $\smooth(T)$ denote the tree obtained from $T$ by smoothing all subdivision vertices, i.e., $\smooth(T)$ is the unique smallest tree that can be subdivided to obtain a tree isomorphic to $T$. We start with an auxiliary lemma which is central in our approach 
\ifArxiv
(see Appendix~\ref{app:trees} for details). 
\else
(see~\cite{bbddgkpr-tror-20-arxiv} for details).
\fi

\begin{restatable}{lemma}{rectilineartrees}\label{lem:rectilinear-trees}
$T$ is turn-regular if and only if $\smooth(T)$ is rectilinear turn-regular.
\end{restatable}
\begin{sketch}
Suppose that $T$ has a turn-regular representation $H$ (the other direction is obvious).  We can assume that $H$ has no zig-zag edges. By Corollary~\ref{cor:outer-face-three-convex}, the rectilinear image of $H$ has at most two consecutive convex corners, which can be removed with local transformations as in Fig.~\ref{fig:rectilinear-trees}. This results 
in a turn-regular representation with only flat corners at degree-2 vertices as desired.       
\end{sketch}

\begin{figure}[tb]
	\centering
	\subfigure[]{%
		\centering
		\includegraphics[page=2]{figures/rectilinear-trees.pdf}
		\label{fig:rectilinear-trees-1-a}
	}\hfill
	\subfigure[]{%
		\centering
		\includegraphics[page=1]{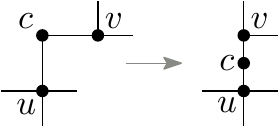}
		\label{fig:rectilinear-trees-1-b}
	}\hfill
	\subfigure[]{%
		\centering
		\includegraphics[page=3]{figures/rectilinear-trees.pdf}
		\label{fig:rectilinear-trees-2}
	}
	\caption{Illustrations for the proof of Lemma~\ref{lem:rectilinear-trees}.}
	\label{fig:rectilinear-trees}
\end{figure}

Unless otherwise specified, from now on we will assume by Lemma~\ref{lem:rectilinear-trees} that~$T$ is a tree without degree-$2$ vertices. We will further refer to a tree as turn-regular if and only if it is rectilinear turn-regular. This implies that the class of turn-regular trees coincides with the class of trees admitting planar straight-line \emph{convex drawings}, i.e., all edges incident to leaves can be extended to infinite crossing-free rays, whose edges are horizontal or vertical segments. 
%
%
The next property directly follows from Lemma~\ref{lem:rectilinear-trees} and the absence of degree-$2$ vertices. 


\begin{property}\label{prp:reflex-leaves}
Let $H$ be a turn-regular rectilinear representation of a tree $T$. Then, the reflex corners of $H$ are formed by the leaves of $T$.
\end{property}

\noindent While turn-regularity is not a hereditary property in general graphs, the next lemma, 
\ifArxiv
whose proof can be found in Appendix~\ref{app:trees}, 
\else
whose proof can be found in~\cite{bbddgkpr-tror-20-arxiv}, 
\fi
shows that it is in fact hereditary~for~trees.

\begin{restatable}{lemma}{hereditary}\label{lem:subtrees}
If a tree $T$ is turn-regular, then any subtree of $T$ is turn-regular.
\end{restatable}

%
%

A \emph{trivial tree} is a single edge; otherwise, it is \emph{non-trivial}. 
For $k\in\{2,3\}$, a \emph{$k$-fork} in a tree $T$ consists of a vertex $v$ whose degree is $k+1$ and at least $k$ leaves adjacent to it in $T$.

For $k\in\{2,3\}$, a \emph{$k$-fork} at a vertex $v$ in a tree $T$ consists of vertex $v$ whose degree is $k+1$ and at least $k$ leaves adjacent to it in $T$.
Due to the degree restriction, a $2$-fork is not a $3$-fork, and vice versa. Also, notice that by definition $K_{1,4}$ is a $3$-fork. 
The next lemma follows from~\cite[Lem.~7]{DBLP:conf/gd/CarlsonE06}; a simplified proof is given 
\ifArxiv
in Appendix~\ref{app:trees}. 
\else
in~\cite{bbddgkpr-tror-20-arxiv}.
\fi

\begin{restatable}
{lemma}{treesforks}
\label{lem:trees-fork}
A turn-regular tree has  
\begin{inparaenum}[(i)]
\item at most four $2$-forks and no $3$-fork, or
\item two $3$-forks and no $2$-fork, or
\item one $3$-fork and at most two~$2$-forks.
\end{inparaenum}
\end{restatable}

\begin{lemma}\label{lem:nontrivial-subtree}
A non-trivial tree $T$ contains at least one $2$- or $3$-fork. 
\end{lemma}
\begin{proof}
Since $T$ is non-trivial and contains no vertices of degree two, there exists a non-leaf vertex $v$ with degree either three or four, such that $v$ is adjacent to exactly two or three leaves, respectively. Thus, the claim follows.
\end{proof}


\begin{corollary}\label{cor:atmost4subtrees}
A turn-regular tree has at most four non-trivial disjoint subtrees.
\end{corollary}
A vertex~$v$ of a tree~$T$ is a \emph{splitter} if $v$ is adjacent to at least three~non-leaf~vertices.
\begin{restatable}{lemma}{lemmaxsplitters}\label{lem:max-splitters}
A turn-regular tree $T$ contains at most two splitters.
\end{restatable}
\begin{proof}
Assume to the contrary that $T$ contains at least three splitters $v_1$, $v_2$ and $v_3$. We first claim that it is not a loss of generality to assume that $v_1$, $v_2$ and $v_3$ appear on a path in $T$. If this is not the case, then there is a vertex, say $u$, such that $v_1$, $v_2$ and $v_3$ lie in three distinct subtrees rooted at $u$. Hence, $u$ is a splitter that lies on the path from $v_1$ to $v_3$. If we choose $v_2$ to be $u$, the claim follows. 

Let $P$ be the path containing $v_1$, $v_2$ and $v_3$ in $T$, and assume w.l.o.g.\ that $v_1$ and $v_3$ are the two end-vertices of $P$. Since $v_1$ is a splitter, it is adjacent to at least three vertices that are not leaves and two of them do not belong to $P$. Let $T_1$ and $T_2$ be the subtrees of $T$ rooted at these two vertices, which by definition are non-trivial and do not contain $v_2$ and $v_3$. By a symmetric argument on $v_3$, we obtain two non-trivial subtrees $T_3$ and $T_4$ of $T$ that do not contain $v_1$ and $v_2$. The third splitter $v_2$ may have only one neighbor that is not a leaf and does not belong to $P$. The (non-trivial) subtree $T_5$ rooted at this vertex contains neither $v_1$ nor $v_3$. Hence, $T_1,\ldots,T_5$ contradict Corollary~\ref{cor:atmost4subtrees}.
\end{proof}

By Lemma~\ref{lem:max-splitters}, a turn-regular tree contains either zero  or one or two splitters (see Lemmas~\ref{lem:0-splitters}-\ref{lem:2-splitters}). A \emph{caterpillar} is a tree, whose leaves are within unit distance from a path, called \emph{spine}. For $k\in\{3,4\}$, a \emph{$k$-caterpillar} is a non-trivial caterpillar (i.e., not a single edge), whose spine vertices have degree at least $3$ and at most~$k$. 

\begin{figure}[t]
  \centering
  \subfigure[]{%
    \centering
    \includegraphics[page=1]{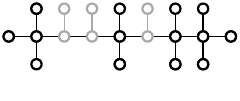}
    \label{fig:4-caterpillar}
  }
  \hfil
  \subfigure[]{%
    \centering
    \includegraphics[page=3]{figures/tree-deg4}
    \label{fig:4-windmill}
  }\hfil
  \subfigure[]{%
    \centering
    \includegraphics[page=4]{figures/tree-deg4}
    \label{fig:3-windmill}
  }\hfil
  \subfigure[]{%
    \centering
    \includegraphics[page=2]{figures/tree-deg4}
    \label{fig:double-windmill}
  }
  \caption{Illustration of
  (a) a $4$-caterpillar,
  (b) a $4$-windmill,
  (c) a $3$-windmill, and
  (d) a double-windmill. Possible extensions are highlighted in gray.}
  \label{fig:characterization}
\end{figure}

\begin{lemma}\label{lem:0-splitters}
A tree $T$ without splitters is a $4$-caterpillar and turn-regular.
\end{lemma}
\begin{proof}
In the absence of splitters in $T$, all inner vertices of $T$ form a path. Hence, $T$ is a $4$-caterpillar and thus turn-regular; see Fig.~\ref{fig:4-caterpillar}. 
\end{proof}

A tree with one splitter $v$ is
\begin{inparaenum}[(i)]
\item a \emph{$4$-windmill}, if $v$ is the root of four $3$-caterpillars (Fig.~\ref{fig:4-windmill}), 
\item a \emph{$3$-windmill}, if $v$ is the root of two $3$-caterpillars and one $4$-caterpillar (Fig.~\ref{fig:3-windmill}).
\end{inparaenum}
Note that in the latter case, $v$ can be adjacent to a leaf if it has degree four.
%
The operation of \emph{pruning} a rooted tree $T$ \emph{at} a degree-$k$ vertex $v$ with $k \in \{3,4\}$ that is not the root of $T$, removes the $k-1$ subtrees of $T$ rooted at the children of $v$ without removing these children, and yields a new subtree $T'$ of $T$, in which $v$ and its children form a $(k-1)$-fork~in~$T'$.

\begin{restatable}{lemma}{lemonesplitter}\label{lem:1-splitter}
A tree $T$ with one splitter is turn-regular if and only if it is a $3$- or $4$-windmill.
\end{restatable}
\begin{sketch}
Every $3$- or $4$-windmill is turn-regular; see Figs.~\ref{fig:4-windmill}-\ref{fig:3-windmill}. Now, let $u$ be the splitter of $T$. If $u$ has four non-leaf neighbors, then $u$ is the root of four non-trivial subtrees $T_1,\ldots,T_4$, which by Lemma~\ref{lem:0-splitters} are $4$-caterpillars. We claim that none of them has a degree-$4$ vertex. Assume to the contrary that $T_1$ contains such a vertex $v \neq u$. We root $T$ at $u$ and prune at $v$, resulting in a (turn-regular, by Lemma~\ref{lem:subtrees}) subtree $T'$ of $T$ that contains a $3$-fork at $v$. By Lemma~\ref{lem:nontrivial-subtree}, each of the non-trivial trees $T_2,\ldots,T_4$ contains a fork. By Lemma~\ref{lem:trees-fork}, these three forks together with the $3$-fork formed at $v$ contradict the turn-regularity of $T'$. The case in which $u$ has three non-leaf neighbors
\ifArxiv
can be found in Appendix~\ref{app:trees}.
\else
can be found in~\cite{bbddgkpr-tror-20-arxiv}.
\fi
\end{sketch}

A tree $T$ with exactly two splitters $u$ and $v$ is a \emph{double-windmill} if 
\begin{inparaenum}[(i)]
\item \label{p:1} the path from $u$ to $v$ in~$T$ forms the spine of a $4$-caterpillar in $T$, 
\item \label{p:2} each of $u$ and $v$ is the root of exactly three non-trivial subtrees, and
\item \label{p:3} the two non-trivial subtrees rooted at $u$ ($v$) that do not contain $v$ ($u$) are $3$-caterpillars; see Fig.~\ref{fig:double-windmill}.
\end{inparaenum}
The proof of the next lemma is similar to the one of Lemma~\ref{lem:1-splitter};
\ifArxiv
see Appendix~\ref{app:trees}.
\else
see~\cite{bbddgkpr-tror-20-arxiv}.
\fi

\begin{restatable}{lemma}{lemtwosplitters}\label{lem:2-splitters}
A tree $T$ with two splitters is turn-regular if and only if it is a double-windmill.
\end{restatable}

\noindent Lemmas~\ref{lem:0-splitters}-\ref{lem:2-splitters} imply the next theorem. Note that for the recognition,~one can test if a (sub-)tree is a $3$- or a $4$-caterpillar in linear time (for details,
\ifArxiv
see~Appendix~\ref{app:trees}).
\else
see~\cite{bbddgkpr-tror-20-arxiv}).
\fi
\begin{restatable}{theorem}{windmills}\label{thml:windmills}
A tree $T$ is turn-regular if and only if $\smooth{(T)}$ is
  \begin{inparaenum}[(i)]
  \item a $4$-caterpillar, or
  \item a $3$- or a $4$-windmill, or
  \item a double-windmill.
  \end{inparaenum}
Moreover, recognition and drawing can be done in linear time. 
\end{restatable}

\section{Turn-Regular Rectilinear Representations }\label{sec:rect}
Here we focus on rectilinear planar representations and prove the following.

\begin{restatable}{theorem}{thsmallfaces}\label{th:small-faces}
  Let $G$ be an $n$-vertex biconnected plane graph with faces of degree at most~eight. There exists an $O(n^{1.5})$-time algorithm that decides whether $G$ admits an embedding-preserving turn-regular rectilinear representation and that computes such a representation in the positive case. 
\end{restatable}
\begin{proof}
We describe a testing algorithm based on a constrained version of Tamassia's flow network $N(G)$, which models the space of orthogonal representations of $G$ within its given planar embedding~\cite{DBLP:journals/siamcomp/Tamassia87}.
Let $V$, $E$, and $F$ be the set of vertices, edges, and faces of $G$, respectively. Tamassia's flow network $N(G)$ is a directed multigraph having a \emph{vertex-node} $\nu_v$ for each vertex $v \in V$ and a \emph{face-node} $\nu_f$ for each face $f \in F$. $N(G)$ has two types of edges: $(i)$ for each vertex $v$ of a face $f$, there is a directed edge $(\nu_v,\nu_f)$ with capacity~3; $(ii)$ for each edge $e \in E$, denoted by $f$ and $g$ the two faces incident to $e$, there is a directed edge $(\nu_f,\nu_g)$ and a directed edge $(\nu_g,\nu_f)$, both with infinite capacity.

\begin{figure}[tb]
	\centering
	\subfigure[]{%
		\centering
		\includegraphics[page=1,width=0.37\textwidth]{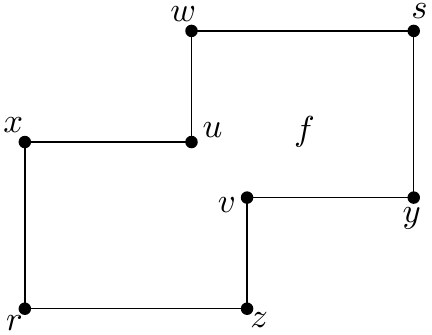}
		\label{fig:small-faces-1}
	}\hfil
	\subfigure[]{%
		\centering
		\includegraphics[page=2,width=0.4\textwidth]{figures/small-faces.pdf}
		\label{fig:small-faces-2}
	}
	\caption{(a) A pair of kitty corners in a face of degree eight. (b) The modification of the flow network around a face-node corresponding to an internal face. The labels on the directed edge represent capacities.}
	\label{fig:small-faces}
\end{figure}

A feasible flow on $N(G)$ corresponds to an orthogonal representation of $G$: a flow value $k \in \{1,2,3\}$ on an edge $(\nu_v,\nu_f)$ represents an angle of $90 \cdot k$ degrees at $v$ in $f$ (since $G$ is biconnected, there is no angle larger than $270^\circ$ at a vertex); a flow value $k \geq 0$ on an edge $(\nu_f,\nu_g)$ represents $k$ bends on the edge of $G$ associated with $(\nu_f,\nu_g)$, and all these bends form an angle of $90^\circ$ inside $f$. 
Hence, each vertex-node $\nu_v$ supplies 4 units of flow in $N(G)$, and each face-node $\nu_f$ in $N(G)$ demands an amount of flow equal to $c_f = (2\deg(f)-4)$ if $f$ is internal and to $c_f = (2\deg(f)+4)$ if $f$ is external. The value $c_f$ represents the \emph{capacity of} $f$.  It is proved in~\cite{DBLP:journals/siamcomp/Tamassia87} that the total flow supplied by the vertex-nodes equals the total flow demanded by the face-nodes; if a face-node $\nu_f$ cannot consume all the flow supplied by its adjacent vertex-nodes (because its capacity $c_f$ is smaller), it can send the exceeding flow to an adjacent face-node $\nu_g$, through an edge $(\nu_f,\nu_g)$, thus originating bends.  

Our algorithm has to test the existence of an orthogonal representation $H$ such that: (a) $H$ has no bend; (b) $H$ is turn-regular. To this aim, we suitably modify $N(G)$ so that the possible feasible flows only model the set of orthogonal representations that verify Properties~(a) and~(b). To enforce Property~(a), we just remove from $N(G)$ the edges between face-nodes. To enforce Property~(b), we enhance $N(G)$ with additional nodes and edges. Consider first an internal face $f$ of $G$. By hypothesis $\deg(f) \leq 8$. It is immediate to see that if $\deg(f) \leq 7$ then $f$ cannot have a pair of kitty corners. If $\deg(f)=8$, a pair $\{u,v\}$ of kitty corners necessarily requires three vertices along the boundary of $f$ going from $u$ to $v$ (and hence also from $v$ to $u$); see Fig.~\ref{fig:small-faces-1}. Therefore, for such a face $f$, we locally modify $N(G)$ around $\nu_f$ as shown in Fig.~\ref{fig:small-faces-2}. Namely, for each potential pair $\{u,v\}$ of kitty corners, we introduce an intermediate node $\nu_{uv}$; the original edges $(\nu_u,\nu_f)$ and $(\nu_v,\nu_f)$ are replaced by the edges $(\nu_u,\nu_{uv})$ and $(\nu_v,\nu_{uv})$, respectively (each still having capacity 3); finally, an edge $(\nu_{uv},\nu_f)$ with capacity $5$ is inserted, which avoids that $u$ and $v$ form a reflex corner inside $f$ at the same time.  
For the external face $f$, it can be easily seen that a pair of kitty corners is possible only if the face has degree at least 10. Since we are assuming that $\deg(f) \leq 8$, we do not need to apply any local modification to $N(G)$ for the external face.

Hence, a rectilinear turn-regular representation of $G$ corresponds to a feasible flow in the modified version of $N(G)$. Since $N(G)$ can be easily transformed into a sparse unit capacity network, this problem can be solved in $O(n^{1.5})$ time by applying a maximum flow algorithm (the value of the maximum flow must be equal to $4|V|$)~\cite{DBLP:books/daglib/0069809}.
\end{proof}

\section{Open Problems}\label{sec:conclusions}

Our work raises several open problems. 
%
\begin{inparaenum}[(i)]
\item A natural question is if all biconnected planar 4-graphs are turn-regular (not only internally). 
\item While we suspect the existence of non-turn regular biconnected planar 4-graphs, we conjecture that triconnected planar 4-graphs are turn-regular.
\item It would be interesting to extend the result of Theorem~\ref{th:small-faces} to more general classes of plane graphs.
\end{inparaenum}


\bibliographystyle{splncs03}
\bibliography{references}

\ifArxiv

\newpage
\appendix
\section*{Appendix}

\section{Additional Material for Section~\ref*{sec:preliminaries}}
\label{app:preliminaries}
\myparagraph{Planar Graphs and Embeddings.} A \emph{$k$-graph} is a graph with vertex-degree at most $k$. We denote by $\deg(v)$ the degree of a vertex $v$. A \emph{plane graph} is a planar graph with a given planar embedding. Let $G$ be a plane graph and let $f$ be a face of $G$. We always assume that the boundary of $f$ is traversed counterclockwise, if $f$ is an internal face, and clockwise, if $f$ is the external face. Note that if $G$ is not biconnected, an edge may occur twice and a vertex may occur multiple times on the boundary of $f$. The total number of vertices (or edges), counted with their multiplicity, is called the \emph{degree of $f$} and is denoted as $\deg(f)$. 

\myparagraph{Orthogonal Drawings and Representations.}
Let $G$ be a planar 4-graph. A planar \emph{orthogonal drawing} $\Gamma$ of $G$ is a planar drawing of $G$ that represents~each vertex as a point and each edge as an alternating sequence of horizontal and vertical segments between its end-vertices. A \emph{bend} in $\Gamma$ is a point of an edge, in which a horizontal and a vertical segment meet. Informally speaking, an orthogonal representation of $G$ is an equivalence class of orthogonal drawings of $G$ having the same planar embedding and the same ``shape'', i.e., the same sequences of angles around the vertices ,and of bends along the edges. 

More formally, if $G$ is plane, and $e_1$ and $e_2$ are two (possibly coincident) edges incident to a vertex $v$ of $G$ that are consecutive in the clockwise order around $v$, we say that $a = \langle e_1,v,e_2 \rangle$ is an \emph{angle at $v$} of $G$ or simply an \emph{angle} of $G$. Let $\Gamma$ and $\Gamma'$ be two embedding-preserving orthogonal drawings of $G$. We say that $\Gamma$ and $\Gamma'$ are \emph{equivalent} if: 
\begin{inparaenum}[(i)]
\item for any angle $a$ of $G$, the geometric angle corresponding to $a$ is the same in $\Gamma$ and $\Gamma'$, and 
\item for any edge $e=(u,v)$ of $G$, the sequence of left and right bends along $e$ moving from $u$ to $v$ is the same in $\Gamma$ and in $\Gamma'$. 
\end{inparaenum}
An \emph{orthogonal representation} $H$ of $G$ is a class of equivalent orthogonal drawings of $G$. Representation $H$ is completely described by the embedding of $G$, by the value $\alpha \in \{90^\circ,180^\circ,270^\circ,360^\circ\}$ for each angle $a$ of $G$ ($\alpha$ defines the geometric angle associated with $a$), and by the ordered sequence of left and right bends along each edge $(u,v)$, moving from $u$ to $v$; if we move from $v$ to $u$ this sequence and the direction (left/right) of each bend are reversed.
An orthogonal representation without bends is also called \emph{rectilinear}. 

W.l.o.g., from now on we assume that an orthogonal representation $H$ comes with a given orientation, i.e., we shall assume that for each edge segment $\rect{pq}$ of $H$ (where $p$ and $q$ correspond to vertices or bends), it is fixed if $p$ is to the left, to the right, above, or below $q$ in every orthogonal drawing that preserves $H$.

\kittycorners*

\begin{proof}
  Let~$\pi$ denote the path from $c_1$ to $c_2$ in a clockwise
  traversal of the external face.  First assume that~$c_1$ is not a
  reflex corner and~$\rot(c_1,c_2) \ge 3$.  Let~$c$ be the corner that
  precedes~$c_1$.  If~$c$ is not a reflex corner,
  then~$\rot(c,c_2) = \rot(c_1,c_2) + \turn(c) \ge 3$, and if~$c$ is a reflex corner,
  then~$\rot(c,c_2) = \rot(c_1,c_2) - 1 \ge 2$.  We can hence iteratively expand the
  path~$\pi$ by adding preceding corners until we find a reflex
  corner~$c_1$ such that~$\rot(c_1,c_2) \ge 2$.
  If~$c_2$ is a reflex corner, and $\rot(c_1,c_2) = 2$, we have kitty
  corners by definition.  If~$\rot(c_1,c_2) > 2$, then, as
  $\rot(c_1,c_1) = -4$ (by Property~\ref{pr:rot-1}) and we only reduce
  the rotation value at reflex corners, we can keep extending~$\pi$
  until we find a pair of corners~$(c_1,c_2')$ with
  $\rot(c_1,c_2') = 2$ such that $c_2'$ is a reflex corner.
  Similarly, if $\rot(c_1,c_2) = 2$ but $c_2$ is not a reflex corner,
  we can extend~$\pi$ to the next reflex corner $c_2'$ on the external
  face with $\rot(c_1,c_2') = 2$.  Then,~$(c_1,c_2')$ is the claimed
  pair of kitty corners.
\end{proof}


\section{Full Proofs for Section~\ref*{sec:turn-regular-graphs}}
\label{app:turn-regular}

\begin{figure}[t!]
  \centering
  \subfigure[]{%
    \centering
    \includegraphics[page=5,width=0.2\textwidth]{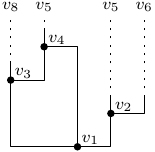}
    \label{fig:whole-example-biconn-deg3-v4}
  }\hfil
  \subfigure[]{%
    \centering
    \includegraphics[page=6,width=0.2\textwidth]{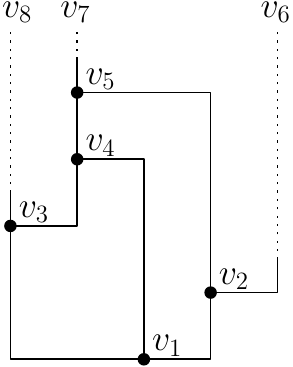}
    \label{fig:whole-example-biconn-deg3-v5}
  }\hfil
  \subfigure[]{%
    \centering
    \includegraphics[page=7,width=0.2\textwidth]{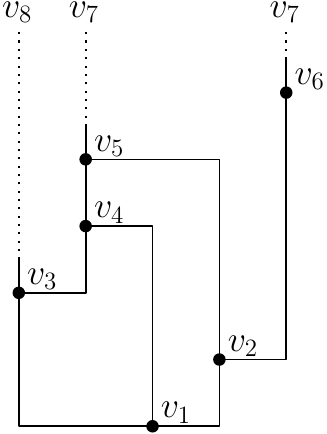}
    \label{fig:whole-example-biconn-deg3-v6}
  }\hfil
  \subfigure[]{%
    \centering
    \includegraphics[page=8,width=0.2\textwidth]{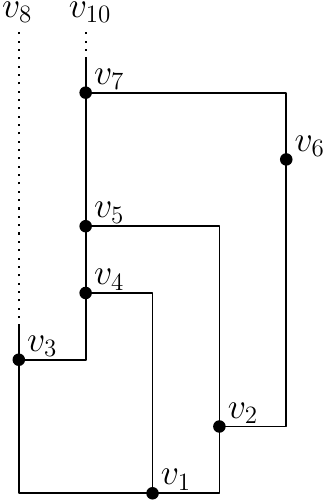}
    \label{fig:whole-example-biconn-deg3-v7}
  }\hfil
  \subfigure[]{%
    \centering
    \includegraphics[page=9,width=0.25\textwidth]{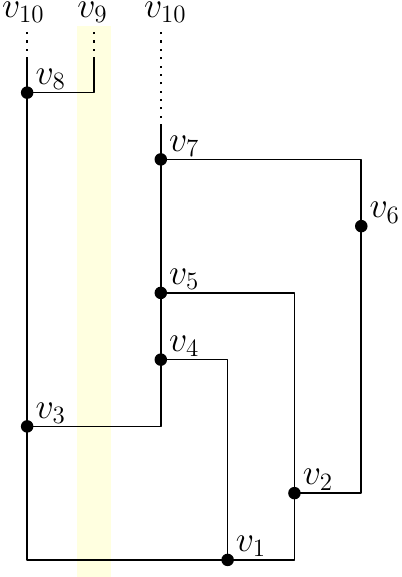}
    \label{fig:whole-example-biconn-deg3-v8}
  }\hfil
  \subfigure[]{%
    \centering
    \includegraphics[page=10,width=0.25\textwidth]{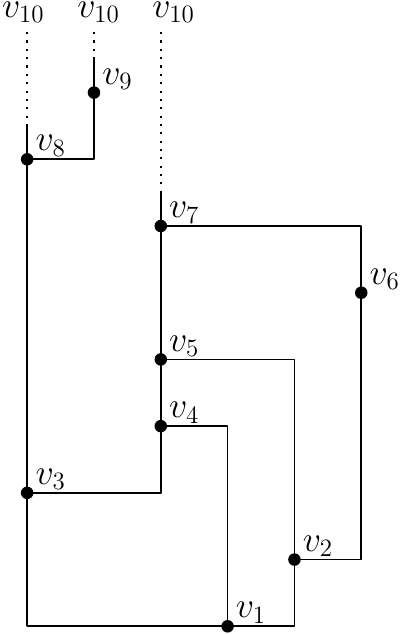}
    \label{fig:whole-example-biconn-deg3-v9}
  }\hfil
  \subfigure[]{%
    \centering
    \includegraphics[page=11,width=0.25\textwidth]{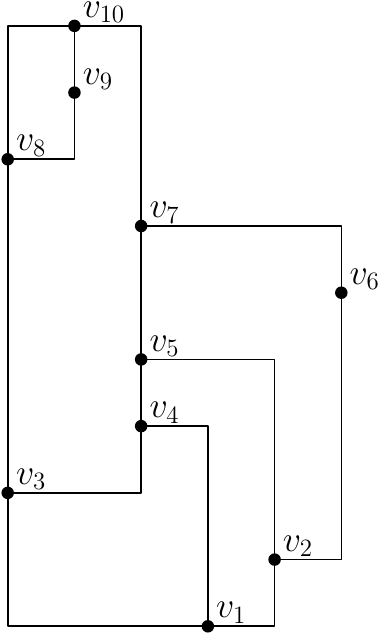}
    \label{fig:whole-example-biconn-deg3-v10}
  }
  \caption{The final steps of the algorithm in the proof of Theorem~\ref{th:biconn-deg3} for the construction of a turn-regular orthogonal drawing of a biconnected planar 3-graph (see also Fig.~\ref{fig:biconn-deg3}).}
  \label{fig:biconn-deg3-whole}
\end{figure}

\hamiltoniandegfour*

\begin{figure}[h!tb]
  \centering
  \subfigure[]{%
    \centering
    \includegraphics[page=1,width=0.15\textwidth]{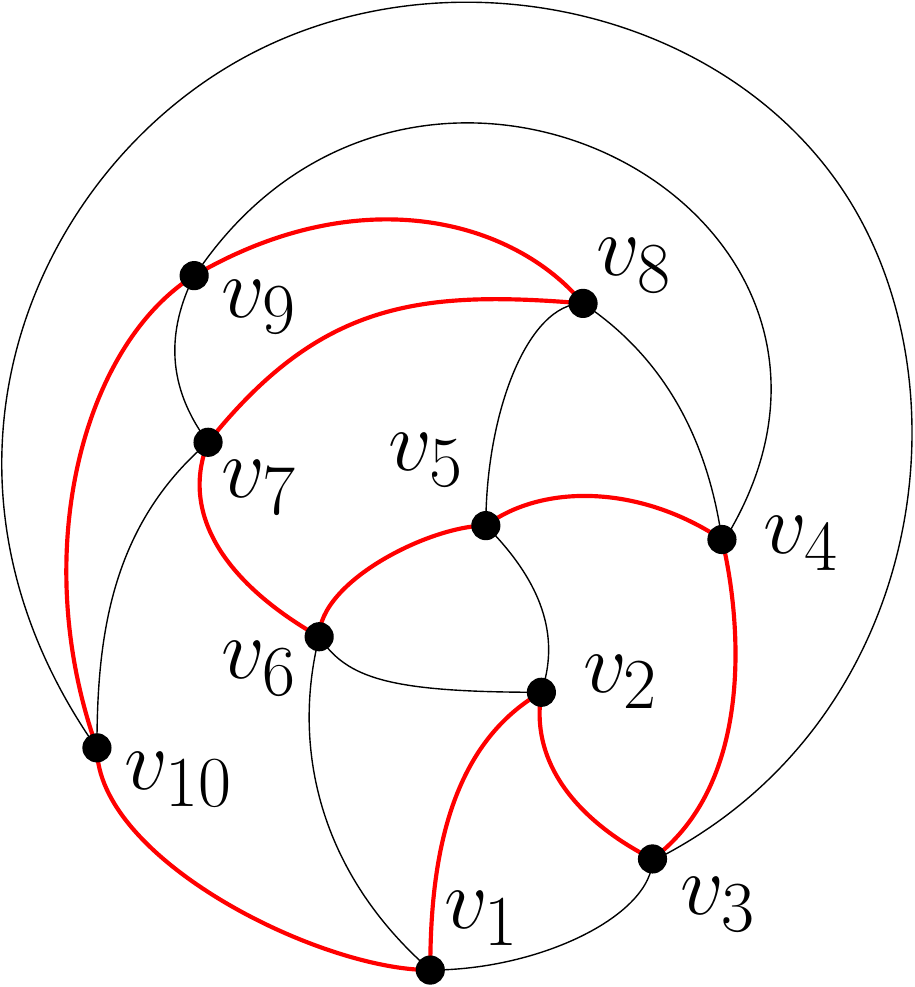}
    \label{fig:example-ham-deg4}
  }\hfil
  \subfigure[]{%
    \centering
    \includegraphics[page=2,width=0.15\textwidth]{figures/hamiltonian-deg4.pdf}
    \label{fig:example-ham-deg4-v1}
  }\hfil
  \subfigure[]{%
    \centering
    \includegraphics[page=3,width=0.20\textwidth]{figures/hamiltonian-deg4.pdf}
    \label{fig:example-ham-deg4-v2}
  }\hfil
  \subfigure[]{%
    \centering
    \includegraphics[page=4,width=0.20\textwidth]{figures/hamiltonian-deg4.pdf}
    \label{fig:example-ham-deg4-v3}
  }\hfil
  \subfigure[]{%
    \centering
    \includegraphics[page=5,width=0.20\textwidth]{figures/hamiltonian-deg4.pdf}
    \label{fig:whole-example-ham-deg4-v4}
  }\hfil
  \subfigure[]{%
    \centering
    \includegraphics[page=6,width=0.24\textwidth]{figures/hamiltonian-deg4.pdf}
    \label{fig:whole-example-ham-deg4-v5}
  }\hfil
  \subfigure[]{%
    \centering
    \includegraphics[page=7,width=0.24\textwidth]{figures/hamiltonian-deg4.pdf}
    \label{fig:whole-example-ham-deg4-v6}
  }\hfil
  \subfigure[]{%
    \centering
    \includegraphics[page=8,width=0.30\textwidth]{figures/hamiltonian-deg4.pdf}
    \label{fig:whole-example-ham-deg4-v7}
  }\hfil
  \subfigure[]{%
    \centering
    \includegraphics[page=9,width=0.30\textwidth]{figures/hamiltonian-deg4.pdf}
    \label{fig:whole-example-ham-deg4-v8}
  }\hfil
  \subfigure[]{%
    \centering
    \includegraphics[page=10,width=0.30\textwidth]{figures/hamiltonian-deg4.pdf}
    \label{fig:whole-example-ham-deg4-v9}
  }\hfil
  \subfigure[]{%
    \centering
    \includegraphics[page=11,width=0.29\textwidth]{figures/hamiltonian-deg4.pdf}
    \label{fig:whole-example-ham-deg4-v10}
  }
  \caption{The drawing produced by the algorithm of the proof of Theorem~\ref{th:hamiltonian-deg4} for the construction of a turn-regular orthogonal drawing of a Hamiltonian planar 4-graph. The Hamiltonian cycle is drawn red and thick.}
  \label{fig:hamiltonian-deg4-whole}
\end{figure}

\begin{proof}
We use an iterative construction inspired by the algorithm by Biedl and Kant~\cite{DBLP:journals/comgeo/BiedlK98} where we replace the $st$-ordering of the input graph with the ordering given by the Hamiltonian cycle.
Let $G$ be a Hamiltonian planar 4-graph and let $\mathcal{E}$ be a planar embedding of $G$. If $G$ has some vertices of degree less than four, choose one of these to be $v_1$. Otherwise, if $G$ is 4-regular, denote by $v_1$ a vertex of $G$ such that the two edges of the Hamiltonian path incident to $v_1$ are not both on the external face. Such a vertex always exists since if $G$ is 4-regular it cannot be also outerplanar, as an outerplanar graph has at least a vertex of degree two. We consider the vertices in the order $v_1, v_2, \dots, v_n$ given by the Hamiltonian cycle of~$G$. We also assume that edge $(v_n,v_1)$ is incident to the external face of $\mathcal{E}$ (otherwise we could change the external face preserving the rotation scheme of $\mathcal{E}$ and avoiding that $(v_1,v_2)$ is also on the external face). We incrementally construct an orthogonal drawing $\Gamma$ of $G$ by adding $v_k$, for $k = 1, \dots, n$, to the drawing $\Gamma_{k-1}$ of $\{v_1, \dots, v_{k-1}\}$, preserving the embedding of $\mathcal{E}$. In particular, vertex $v_k$ is always placed above~$\Gamma_{k-1}$.

After $v_k$ is added to $\Gamma_{k-1}$, we introduce some extra columns into $\Gamma_{k}$ so to maintain some invariants. Analogously to~\cite{DBLP:journals/comgeo/BiedlK98} and to the proof of Theorem~\ref{th:biconn-deg3}, we maintain the invariant that each edge $(v_i,v_j)$ such that $i \leq k < j$ has a dedicated column in $\Gamma_{k}$ that is reachable from $v_i$ with at most one bend without introducing crossings. Observe that, since the vertices are inserted in the order given by the Hamiltonian cycle, in the drawing $\Gamma_{k-1}$ a special path can be identified, that we call the \emph{spine}, connecting, for $i=1, \dots, k-2$, vertex $v_i$ to $v_{i+1}$. We maintain the invariant that all the reflex corners introduced in the drawing point down-left or up-left if they are contained into a face that is on the left side of the spine and point down-right or up-right if they are contained into a face that is contained on the right side of the spine, with the possible exception of the reflex corners on the external face occurring on edges incident to~$s$ or to~$t$.

Suppose $v_1$ has degree $4$. Then, its additional two edges may be: (a) both on the right side of the spine (Fig.~\ref{fig:ham-first-1}) or (b) one on the left side and the other on the right side of the spine (Fig.~\ref{fig:ham-first-2}). The case when the additional edges are both on the left side of the spine, depicted in Fig.~\ref{fig:ham-first-3}, is ruled out by the choice of $v_1$. In the cases (a) and (b) the drawing of $v_1$ and the columns reserved for its outgoing edges are depicted in Figs.~\ref{fig:ham-first-1} and \ref{fig:ham-first-2}, respectively. If, instead, $v_1$ has degree three or two, the drawing of $v_1$ is obtained from Fig.~\ref{fig:ham-first-2} by omitting the missing edges.

For $k=1, 2, \dots, n$, each vertex $v_k$ added to the drawing has one incoming and one outgoing edge of the spine. The drawing of $v_k$ follows simple rules that depend on the number of additional edges of $v_k$ on the left side and on the right side of the spine in~$\mathcal{E}$. 

Now consider a vertex $v_k$ with $1 < k < n$. Suppose $v_k$ has degree four and that its additional edges are both on the left side of the spine (Fig.~\ref{fig:ham-unbalanced-1}). Denote by $v_i$ and $v_j$ the other endpoints of the additional edges of $v_k$ and assume, without loss of generality, that $v_i < v_j$. There are three cases: $v_k < v_i < v_j$ (Fig.~\ref{fig:ham-unbalanced-2}), or $v_i < v_k < v_j$ (Fig.~\ref{fig:ham-unbalanced-3}), or $v_i < v_j < v_k$ (Fig.~\ref{fig:ham-unbalanced-4}). In all cases the drawings of $v_k$ depicted in Figs.~\ref{fig:ham-unbalanced-2}-\ref{fig:ham-unbalanced-4}) guarantee the invariants. 
If $v_k$ has degree two or three, its drawing is that of Fig.~\ref{fig:ham-unbalanced-3} where the missing edges are omitted.

\begin{figure}[!h]
  \centering
  \subfigure[]{%
    \centering
    \includegraphics[page=1,width=0.15\textwidth]{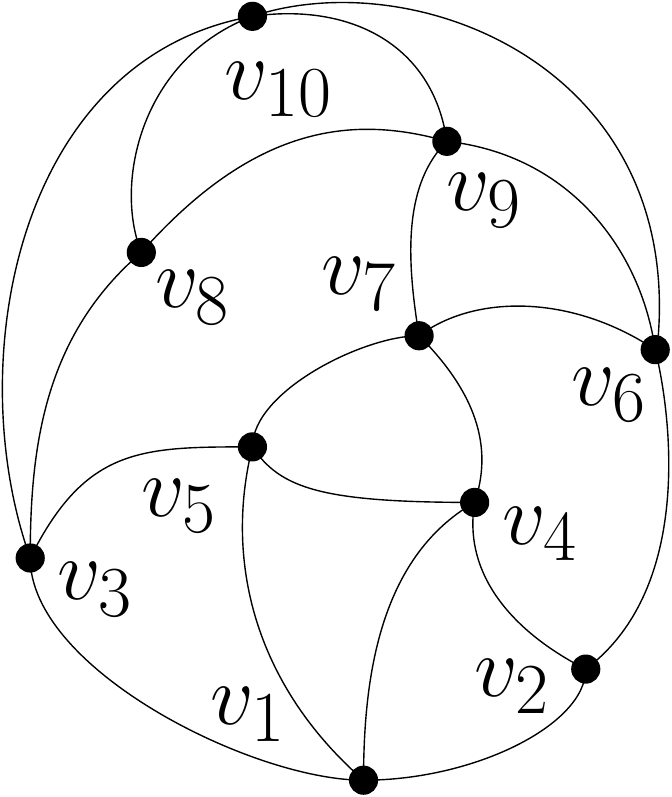}
    \label{fig:roll-whole-instance}
  }\hfil
  \subfigure[]{%
    \centering
    \includegraphics[page=2,width=0.05\textwidth]{figures/roll-whole.pdf}
    \label{fig:roll-whole-v1}
  }\hfil
  \subfigure[]{%
    \centering
    \includegraphics[page=3,width=0.1\textwidth]{figures/roll-whole.pdf}
    \label{fig:roll-whole-v2}
  }\hfil
  \subfigure[]{%
    \centering
    \includegraphics[page=4,width=0.15\textwidth]{figures/roll-whole.pdf}
    \label{fig:roll-whole-v3}
  }\hfil
  \subfigure[]{%
    \centering
    \includegraphics[page=5,width=0.15\textwidth]{figures/roll-whole.pdf}
    \label{fig:roll-whole-v4}
  }\hfil
  \subfigure[]{%
    \centering
    \includegraphics[page=6,width=0.24\textwidth]{figures/roll-whole.pdf}
    \label{fig:roll-whole-v5}
  }\hfil
  \subfigure[]{%
    \centering
    \includegraphics[page=7,width=0.27\textwidth]{figures/roll-whole.pdf}
    \label{fig:roll-whole-v6}
  }\hfil
  \subfigure[]{%
    \centering
    \includegraphics[page=8,width=0.33\textwidth]{figures/roll-whole.pdf}
    \label{fig:roll-whole-v7}
  }\hfil
  \subfigure[]{%
    \centering
    \includegraphics[page=9,width=0.33\textwidth]{figures/roll-whole.pdf}
    \label{fig:roll-whole-v8}
  }\hfil
  \subfigure[]{%
    \centering
    \includegraphics[page=10,width=0.39\textwidth]{figures/roll-whole.pdf}
    \label{fig:roll-whole-v9}
  }\hfil
  \subfigure[]{%
    \centering
    \includegraphics[page=11,width=0.39\textwidth]{figures/roll-whole.pdf}
    \label{fig:roll-whole-v10}
  }
  \caption{Example illustrating the algorithm in the proof of Theorem~\ref{th:biconnected-internally} for the construction of an orthogonal drawing with turn-regular internal faces.}
  \label{fig:roll-whole}
\end{figure}

The case when $v_k$ has degree four and its additional edges are both on the right of the spine can be handled analogously, horizontally mirroring the configurations of Figs.~\ref{fig:ham-unbalanced-2}-\ref{fig:ham-unbalanced-4}. 

Suppose now that $v_k$ has degree four and that its additional edges are one on the left side and one on the right side of the spine (Fig.~\ref{fig:ham-balanced-1}). Again, depending on whether the vertices adjacent to $v_k$ precede or follow $v_k$ in the ordering given by the Hamiltonian cycle we can adopt for $v_k$ one of the drawings depicted in Figs.~\ref{fig:ham-balanced-2}-\ref{fig:ham-balanced-4}.  

Finally, vertex $v_n$ is added to the drawing by using one of the configurations shown in Figs.~\ref{fig:ham-last-1}, \ref{fig:ham-last-2}, and \ref{fig:ham-last-3}.

Since the configurations in Figs~\ref{fig:ham-first-1}-\ref{fig:ham-balanced-4} all respect the invariant that in the faces on the left side (right side, resp.) of the spine only top-left and bottom-left (top-right and bottom-right, resp.) reflex corners are inserted, the drawing cannot have kitty corners and is turn regular. Also observe that each edge has a maximum of three bends per edge.
\end{proof}



\section{Full Proofs for Section~\ref*{sec:testing-turn-regularity}}
\label{app:trees}
In this section, we provide the detailed proofs of statements omitted in Section~\ref{sec:testing-turn-regularity} due to space constraints.

\rectilineartrees*

\begin{proof}
Suppose that $T$ is turn-regular, and let $H$ be a turn-regular representation of $T$ (the other direction is obvious). Consider the rectilinear image $\rect{H}$ of $H$. We show that $\rect{H}$ can be transformed into a turn-regular representation $\rect{H'}$ with only flat corners at degree-2 vertices. Let $u$ and $v$ be two vertices of $\rect{H}$ such that $\deg(u) \neq 2$, $\deg(v) \neq 2$, and the path $\pi_{uv}$ connecting $u$ to $v$ in $\rect{H}$ has only (possibly none) degree-2 vertices. W.l.o.g., we assume that $\pi_{uv}$ does not contain two consecutive corners $c_1$ and $c_2$ such that $\turn(c_1)=1$ and $\turn(c_2)=-1$ (i.e., corners that form a zig-zag pattern), because in this case we can replace both $c_1$ and $c_2$ with two flat corners. Hence, we can assume that all non-flat corners encountered along $\pi_{uv}$, while moving from $u$ to $v$, always turn in the same direction, say to the right; i.e., all of them are convex corners. Let~$k$ be the number of convex corners along $\pi_{uv}$. Since the only face of $H$ is the external face, by Corollary~\ref{cor:outer-face-three-convex} we may assume that  $1 \leq k \leq 2$, which yields two cases. 

Assume first $k=1$. Let $c$ be the convex corner along $\pi_{uv}$. W.l.o.g., suppose that $c$ points up-left. Since $\rect{H}$ is turn-regular, by Corollary~\ref{cor:outer-face-three-convex}, either $u$ has no edge segment incident from the right (see Fig.~\ref{fig:rectilinear-trees-1-a}) or $v$ has no edge segment incident from below (see Fig.~\ref{fig:rectilinear-trees-1-b}). Hence, we can transform $c$ into a flat corner by using one of these two directions to reroute $\pi_{uv}$ around $u$ or around $v$.


Let now $k=2$. Let $c_1$ and $c_2$ be the two convex corners along $\pi_{uv}$. W.l.o.g., suppose that $c_1$ points up-left and $c_2$ points up-right (see Fig.~\ref{fig:rectilinear-trees-2}). Since $\rect{H}$ is turn-regular, by Corollary~\ref{cor:outer-face-three-convex}, $u$ has no edge segment incident from the right and $v$ has no edge segment incident from the left. Hence, again, we can transform $c$ into a flat corner by using these two directions to reroute $\pi_{uv}$ around $u$ and $v$.    	

$\rect{H'}$ is obtained by applying the above transformation to each pair of vertices $u$ and $v$ that have the properties above. 
\end{proof}

\hereditary*

\begin{proof}
Let $H$ be a turn-regular rectilinear representation of tree $T$, and let $T'$ be a subtree of $T$, i.e., $T'$ is a connected subgraph of $T$. If $T$ and $T'$ consist of $n$ and $n'$ vertices, respectively, then $T'$ can be obtained from $T$ by repeatedly removing exactly one leaf of it, $n-n'$ times. For $i=0,1\ldots,n-n'$, denote by $T_i$ the subtree of $T$ that is derived after the removal of the first $i$ leaves, and by $H_i$ the restriction of $H$ to tree $T_i$; clearly, $T = T_0$, $H = H_0$ and $T' = T_{n-n'}$ hold. For $i=1\ldots,n-n'$, we will prove that $H_i$ is turn-regular, under the assumption that $H_{i-1}$ is turn-regular. This will imply that $T'$ is turn-regular, as desired. Let $u$ be the leaf that is removed from $H_{i-1}$ to obtain $H_i$. We emphasize that $H_{i-1}$ consists of a single face, i.e., the external. Hence, the removal of $u$ from $H_{i-1}$ implies the removal of a pair of reflex corners from its external face. 
We next focus on the case, in which the removal of $u$ from $H_{i-1}$ introduces a reflex corner in $H_i$ that is not present in $H_{i-1}$, as otherwise $H_i$ is clearly turn-regular. 

Let $v$ be the (unique) neighbor of $u$ in $T_{i-1}$. If $v$ has no neighbor in $T_{i-1}$ other than $u$, then $i=n-1$ holds, i.e., $T_i$ consists of single vertex $v$ and thus is turn-regular. Hence, we may assume that $v$ has a neighbor, say $w$, in $T_{i-1}$ that is different from $v$. W.l.o.g., we assume that $(v,w)$ is vertical in $H_{i-1}$ (and thus in $H_i$) with $v$ being its top end-vertex. If $\deg(v)=4$ in $T_{i-1}$, then the removal of $u$ from $T_{i-1}$ yields a flat corner in $H_i$, which implies $H_i$ does not contain a reflex corner that does not exist in $H_{i-1}$; a contradiction. Hence, $\deg(v)\in\{2,3\}$ in $T_{i-1}$. We will focus on the case, in which $\deg(v) = 2$ in $T_{i-1}$; the case, in which $\deg(v) = 3$ in $T_{i-1}$, is analogous. Note that if $(u,v)$ is horizontal in $H_{i-1}$, then $(u,v)$ and $(v,w)$ inevitably form a reflex corner $\zeta_v$ in $H_{i-1}$.

Assume for a contradiction that $H_i$ is not turn-regular. If we denote by $\langle c_v, c_v' \rangle$ the ordered pair of reflex corners at $v$ in $H_i$, then at least one of $c_v$ and $c_v'$, say the former, forms a pair of kitty corners with another corner $c$ of $H_i$ (and thus of $H_{i-1}$). Denote by $\langle c_u, c_u' \rangle$ the ordered pair of reflex corners at $u$ in $H_{i-1}$. If $(u,v)$ is horizontal and $u$ is to the left (right) of $v$ in $H_{i-1}$, then we observe that $c$ and $c_u'$ ($c$ and $\zeta_v$, respectively) form a pair of kitty corners in $H_{i-1}$. On the other hand, if $(u,v)$ is vertical, then $c$ and $c_u$ form a pair of kitty corners in $H_{i-1}$. In both cases, we obtain a contradiction to the fact that $H_{i-1}$ is turn-regular.
\end{proof}

\treesforks*

\begin{proof}
To prove the lemma, we first state and formally prove two claims.

\begin{clm}\label{clm:rot-leaves}
Let $H$ be a turn-regular rectilinear representation of a tree $T$. Let $u$ and $v$ be two leaves associated with two ordered reflex corner-pairs $\langle c_u,c_u'\rangle$ and $\langle c_v,c_v'\rangle$. If in a traversal of the external face of $T$ from $u$ to $v$ there is no other leaf of $T$, then $\rot(c_u',c_v)\in\{-1,0,1\}$ (or equivalently $\rot(c_u,c_v)\in\{0,-1,-2\}$).
\end{clm}
\begin{proof}
By Property~\ref{prp:reflex-leaves}, it follows that there exist no reflex corners between $c_u'$ and $c_v$. By Lemma~\ref{lem:outer-face-kitty-corners}, it follows that  $\rot(c_u',c_v)<2$, which implies there exist at most two convex corners between $c_u'$ and $c_v$. Hence, $\rot(c_u',c_v)\in\{-1,0,1\}$.
\end{proof}

\noindent Consider a leaf $u$ of $T$ and assume w.l.o.g.\ that $u$ is pointing \lupward in $H$. Let $v$ be the leaf that follows $u$ in the traversal of the external face of $T$ in $H$. Let also $\langle c_u,c_u'\rangle$ and $\langle c_v,c_v'\rangle$ be the two pairs of ordered reflex corners associated with $u$ and $v$, respectively. By Claim~\ref{clm:rot-leaves}, $v$ does not point \lleftward, as otherwise $\rot(c_u,c_v)=-3$ or $\rot(c_u,c_v)=1$. If $v$ is pointing \lupward (i.e., $\rot(c_u,c_v)=0$), then we say that there exists no \emph{change in direction} between $u$ and $v$. Otherwise, we say that there exists a change in direction between $u$ and $v$, which implies $\rot(c_u,c_v)<0$. 

\begin{clm}\label{clm:not-too-many-directions}
Let $H$ be a turn-regular rectilinear representation of a tree $T$. Then, the total number of changes in direction of the leaves of $T$ in $H$ is at most four.
\end{clm}
\begin{proof}
Assume to the contrary that 
there exist $s \geq 5$ pairs of consecutive leaves $\langle v_1,v_2\rangle,\ldots,\langle v_{2s-1},v_{2s} \rangle$ in a traversal of the external face of $T$, where a change in direction occurs. Note that $v_1,v_2,\ldots,v_{2s-1},v_{2s}$ are not necessarily distinct. For $i=1,\ldots,2s$, let $\langle c_i,c_i'\rangle$ be the ordered  pair of corners associated with $v_i$. Since $\langle v_{2i-1},v_{2i}\rangle$ introduces a change in direction, $\rot(c_{2i-1},c_{2i})\leq-1$. Summing up over $i$, we obtain $\sum_{i=1}^{s} \rot(c_{2i-1},c_{2i}) \leq -5$. Since the remaining leaves of $T$ do not introduce a change in direction,  $\rot(u,u) = \sum_{i=1}^{s} \rot(c_{2i-1},c_{2i})$ holds for any leaf $u$ in $H$, which is a contradiction to Property~\ref{pr:rot-1}.
\end{proof}
Since $2$-~and~$3$-forks require one and two changes in direction, the proof of lemma follows directly from Claim~\ref{clm:not-too-many-directions}.
\end{proof}

%
%

\lemonesplitter*

\begin{proof}
Clearly, every $3$- and $4$-windmill is turn-regular; this can be easily seen looking at the illustrations in Figs.~\ref{fig:3-windmill} and~\ref{fig:4-windmill}, where no pairs of kitty corners are present. For the other direction, consider a turn-regular tree $T$ and let $u$ be the unique splitter of $T$. By definition, $u$ has either four or three neighbors that are not leaves. In the former case, we will prove that $T$ is a $4$-windmill, while in the latter case that $T$ is a $3$-windmill. Note that $u$ may be adjacent to a leaf, only if $\deg(u)=4$. 

We start with the case in which $u$ has four neighbors that are not leaves, which implies that $u$ is the root of exactly four non-trivial subtrees $T_1,\ldots,T_4$. Since $u$ is the only splitter in $T$, by Lemma~\ref{lem:0-splitters}, it follows that each of $T_1,\ldots,T_4$ is a $4$-caterpillar. To show that $T$ is a $4$-windmill, it remains to show that each of $T_1,\ldots,T_4$ cannot contain a degree-$4$ vertex. Assume to the contrary that $T_1$ contains a degree-$4$ vertex $v$. Since the $\deg(u)=1$ in $T_1$, $u \neq v$ holds. We root $T$ at $u$ and we proceed by pruning $T$ at $v$, which will result in a subtree $T'$ of $T$ that contains a $3$-fork at $v$. Since $T$ is turn-regular, by Lemma~\ref{lem:subtrees}, $T'$ is also turn-regular. Since $T_2,\ldots,T_4$ are non-trivial, by Lemma~\ref{lem:nontrivial-subtree} each of them contains a fork. By Lemma~\ref{lem:trees-fork}, these three forks together with the $3$-fork formed at $v$ contradict the fact that $T'$ is turn-regular.

We now consider the case in which $u$ has three neighbors that are not leaves, that is, $u$ is the root of exactly three non-trivial subtrees $T_1$, $T_2$ and $T_3$. Since $u$ is the only splitter in $T$, we have again that each of these subtrees is a $4$-caterpillar. We now claim that two of them must be $3$-caterpillars, which also implies that $T$ is a $3$-windmill, as desired. Assume to the contrary that $T_1$ and $T_2$ are not $3$- but $4$-caterpillars, that is, there exist vertices $v_1$ and $v_2$ in $T_1$ and $T_2$ that are of degree~$4$, respectively. Note that $u \neq v_1$ and $u \neq v_2$. We assume that $T$ is rooted at $u$ and, as in the previous case, we prune $T$ first at $v_1$ and then at $v_2$. The resulting subtree $T'$ contains two $3$-forks at $v_1$ and $v_2$ and one additional fork that is contained in $T_3$ (by Lemma~\ref{lem:nontrivial-subtree}). Hence, by  Lemma~\ref{lem:trees-fork}, $T'$ is not turn-regular, which is a contradiction to Lemma~\ref{lem:nontrivial-subtree} (since $T'$ is a subtree of the turn-regular tree $T$).
\end{proof}

\lemtwosplitters*

\begin{proof}
Clearly, every double-windmill is turn-regular, as it can be easily seen looking at the illustration in Fig.~\ref{fig:double-windmill}. 
Now consider a turn-regular tree $T$ and let $u$ and $v$ be the two splitters of $T$. 

To prove that $T$ has Property~(\ref{p:1}) of a double-windmill, consider the path $P$ from $u$ to $v$ in $T$, and let $w$ be an internal vertex of $P$ (if any). Since $w$ is an internal vertex of $P$, it is adjacent to two non-leaves of $T$ (i.e., its neighbors in $P$). 
Note that $w$ cannot be a splitter, because $T$ has exactly two splitters (that is, $u$ and $v$). Hence, the neighbors of $w$ that are not in $P$ are leaves, which implies that Property~(\ref{p:1}) of a double-windmill holds for $T$.

We now prove that $T$ has Property~(\ref{p:2}) of a double-windmill. Assume to the contrary that $u$ is the root of four non-trivial subtrees. We root $T$ at $u$ and we denote by $T_1,\ldots,T_4$ the subtrees of $T$ that are rooted at the four children of $u$. It follows that $T_1,\ldots,T_4$ are non-trivial and disjoint. W.l.o.g., assume that $T_1$ contains $v$. Since $v$ is a splitter, there exist at least two non-trivial subtrees of $T_1$, say $T_1^1$ and $T_1^2$, that are rooted at two children of $v$. Hence, $T_1^1$, $T_1^2$, $T_2$, $T_3$ and $T_4$ are disjoint subtrees of $T$. Since each of these subtrees is non-trivial, by Corollary~\ref{cor:atmost4subtrees} we have a contradiction to the turn-regularity of $T$. Hence, Property~(\ref{p:2}) of a double-windmill holds for $T$.

By Property~(\ref{p:2}), $u$ has three non-leaf neighbors $u_1,\ldots,u_3$ in $T$. We again root $T$ at $u$ and we denote by $T_1,\ldots,T_3$ the subtrees of $T$ that are rooted at $u_1,\ldots,u_3$. It follows that $T_1,\ldots,T_3$ are non-trivial and disjoint. As above, we assume w.l.o.g.\ that $T_1$ contains $v$, and we define in the same way the two subtrees, $T_1^1$ and $T_1^2$, of $T_1$. We now claim that none of $T_1^1$, $T_1^2$, $T_2$ and $T_3$ contains a degree-$4$ vertex $z$. Assume to the contrary that one, say $T_3$, contains a degree-$4$ vertex $z$. Since we have assumed $T$ to be rooted at $u$, we proceed by pruning $T$ at $z$. The resulting tree $T'$, which is turn-regular by Lemma~\ref{lem:subtrees}, contains a $3$-fork at $z$ and three additional forks that lie in the non-trivial trees $T_1^1$, $T_1^2$ and $T_2$, which contradicts Corollary~\ref{cor:atmost4subtrees}. Hence, our claim follows. In particular, since $u$ and $v$ are the only splitters in $T$, our claim implies that each of $T_2$ and $T_3$ together with $u$, as well as, each of $T_1^1$ and $T_1^2$ together with $v$ is a $3$-caterpillar, which implies that Property~(\ref{p:3}) of a double-windmill holds for $T$.
\end{proof}

\windmills*

\begin{proof}
Since by Lemma~\ref{lem:max-splitters} a turn-regular tree has at most two splitters, the correctness of our characterization follows from Lemmas~\ref{lem:0-splitters}, \ref{lem:1-splitter} and~\ref{lem:2-splitters} (recall that we have assumed w.l.o.g.\ that $T$ does not have degree-2 vertices). For the recognition, we first count how many splitter tree $T$ contains, which can be done in $O(n)$ time. If there are more than two splitters, we reject the instance. If there are no splitters, we accept the instance and we report the representation described in Lemma~\ref{lem:0-splitters}. For the remaining cases, we first observe that one can trivially test whether a (sub-)tree is a $3$- or a $4$-caterpillar in time linear to its number of vertices. This observation directly implies that in linear time one can test whether $T$ is turn-regular, when $T$ has exactly one splitter. It remains to consider the case in which $T$ contains exactly two splitters $u$ and $v$. It follows that each internal vertex of the path from $u$ to $v$ is adjacent only to leaves of $T$. We now argue for vertex $u$; symmetric arguments hold for $v$. Two of the subtrees of $T$ rooted at $u$ that do not contain $v$ have to be $3$-caterpillars, while the third (if any) has to be a leaf. Both can be checked in time linear in the size of $T$, which completes the description of the proof.
\end{proof}


\fi 

\end{document}